\documentclass[runningheads]{llncs}
\usepackage{xcolor}
\usepackage{amsmath,amssymb}

\usepackage{amsthm}
\newcommand{\defn}[1]{\emph{#1}}

\newcommand{\ialphabet}{\Sigma}
\newcommand{\tokenem}{\texttt{em}}
\newcommand{\EOS}{\texttt{EOS}}
\newcommand{\transformer}{\mathcal{T}}

\newcommand{\ACzero}{\text{AC}^0}
\newcommand{\TCzero}{\text{TC}^0}
\newcommand{\NCone}{\text{NC}^1}
\newcommand{\Logspace}{\text{L}}

\newcommand{\parity}{\textsc{PARITY}}

\newcommand{\N}{\mathbb{N}}
\newcommand{\Z}{\mathbb{Z}}

\newcommand{\R}{\mathbb{R}}
\newcommand{\CLTL}{\text{K}_t[\#]}

\newcommand{\vecT}{\textbf{t}}
\newcommand{\vecV}{\textbf{v}}
\newcommand{\vecW}{\textbf{w}}
\newcommand{\vecX}{\textbf{x}}
\newcommand{\vecY}{\textbf{y}}
\newcommand{\vecC}{\textbf{c}}

\newcommand{\LeftCount}[1]{\overleftarrow{\#}[#1]}
\newcommand{\RightCount}[1]{\overrightarrow{\#}[#1]}

\newcommand{\idim}{\texttt{iDim}}
\newcommand{\odim}{\texttt{oDim}}
\newcommand{\weight}{\texttt{wt}}

\newcommand{\softmax}{\texttt{softmax}}
\newcommand{\uha}{\texttt{uha}}
\newcommand{\aha}{\texttt{aha}}

\newif\ifdraft\drafttrue
\ifdraft
\newcommand\al[1]{{\color{blue}
[#1 - \textbf{Anthony}]}}
\newcommand\lo[1]{{\color{red}
[#1 - \textbf{Luke}]}}
\newcommand\ms[1]{{\color{OliveGreen}
[#1 - \textbf{Max}]}}

\else
\newcommand\todo[1]{}
\newcommand\al[1]{}
\newcommand\lo[1]{}
\newcommand\ms[1]{}

\fi

\newcommand{\Global}{\textbf{G}}
\newcommand{\F}{\textbf{F}}
\newcommand{\Mod}{\text{Mod}}
\newcommand{\ReLU}{\text{ReLU}}

\newcommand{\Parikh}{\mathcal{P}}

\usepackage[T1]{fontenc}
\usepackage{graphicx}
\begin{document}
\title{The Role of Logic and Automata in Understanding Transformers}
%
%
\author{
Anthony~W. Lin\inst{1,2}\orcidID{0000-0003-4715-5096} \and Pablo Barcelo\inst{3}\orcidID{0000-0003-2293-2653}}
\authorrunning{Lin and Barcelo}
%
\institute{
Max-Planck Institute for Software Systems, Kaiserslautern, Germany
\and
University of Kaiserslautern-Landau, Kaiserslautern, Germany \and 
Institute for Mathematical and Computational Engineering, Pontificia Universidad
    Cat\'{o}lica de Chile \& IMFD Chile \& CENIA Chile
}
\maketitle              
\begin{abstract}
The advent of transformers has in recent years led to powerful and revolutionary Large Language Models (LLMs). Despite this, our understanding on the capability of transformers is still meager. In this invited contribution, we recount the rapid progress in the last few years to the question of what transformers can do. In particular, we will see the integral role of logic and automata (also with some help from circuit complexity) in answering this question. We also mention several open problems at the intersection of logic, automata, verification and transformers.

\end{abstract}
\section{Introduction}

Recent years witnessed the unprecedented emergence of Large Language Models
(LLMs), which have revolutionized many aspects of our lives. LLMs are based on a
new neural network model called \emph{transformers}, which extends the classical
feed-forward neural network model via \emph{attention mechanisms} for handling
texts of arbitrary lengths. Unlike Recurrent Neural Networks (RNN) \cite{elman} --- which predated transformers by decades --- transformers have proven to be efficiently parallelizable and able to capture long-range dependencies better in practice. Despite the rapid adoption of transformers as a mainstream ML model, some limitations of the transformer model have only been understood in recent years. One good example of such a limitation is to perform \emph{counting} in a text, e.g., determine whether there is an even or an odd number of occurrences of a given token in a text. 

In recent years, subareas of theoretical computer science --- including logic,
automata, and circuit complexity --- have  featured in the rapid development of
the theory of expressivity of transformers (cf. \cite{transformers-survey}).
Such a connection has organically materialized because transformers are computational models that process texts (i.e., strings) and can be studied just like formal models such as finite-state automata, Turing machines, or logics like first-order and second-order logics on strings. Multiple formal models have been developed by varying the following aspects of transformers: attention mechanisms, positional encodings, precision, and the so-called ``chain of thoughts''. Guided by both theory building and experimentation, a picture on the expressive power of transformers has slowly emerged. Although this picture is to date incomplete, a respectable body of works have been produced in the so-called FLaNN (Formal Languages and Neural Networks) community, consisting of logicians, automata theorists, and computational linguists. 

\paragraph{Why this article?} This article has been written to recount \emph{some}
gems that have been discovered at the intersection  of logic, automata, circuit
complexity, and transformers. That is, we do not aim to be exhaustive. The
choices of materials are additionally based on our subjective
taste\footnote{Before working on FLaNN, the authors primarily researched in
logic, automata theory, automated reasoning, finite model theory, and
databases.}. The intended audience of the article includes researchers in logic,
automata, verification and programming languages.  In particular, we will
mention several open problems, which we believe are worth undertaking in the
next years. 

\paragraph{Highlight of key results.} In its simplest form, a transformer can be understood as a formal model that takes an input  \emph{text} (i.e. string) and outputs a {\em token} (i.e. letter). More formally, a transformer gives rise to a function $f: \ialphabet^* \to \ialphabet$, for some finite alphabet $\ialphabet$ of tokens. Moreover, one could think of $f$ as a family of formal languages $\{L_a\}_{a \in \ialphabet}$, where $L_a := \{ w \in \ialphabet^* : f(w) = a\}$. This connection underlines the bridge between formal languages and transformers:  one can simply study such formal languages $L_a$ generated (or recognized) by transformers. 

The first set of results in the paper concerns the expressivity of transformers with {\em unique hard attention} mechanisms (a.k.a. Unique Hard Attention Transformers, or simply UHAT). Such an attention mechanism --- which finds the leftmost value that maximizes the attention score --- is a simplification of \emph{softmax attention}, which is used in practice but has proven to be tricky to analyze in theory owing to the use of such real-valued functions as $e^x$.  
The first key result that we discuss in the paper is from \cite{BKLP24,masked-uhat}. It connects 
formal languages definable in various fragments of first-order logic over strings extended with all numerical predicates (equivalently, subclasses of the circuit complexity class $\ACzero$) and UHAT. In particular, the language
\[
\parity := \{ w \in \{a,b\} : |w|_a \equiv 0 \pmod{2} \}
\]
is well-known \cite{Ajtai83} not to be in $\ACzero$, therefore cannot be expressed by UHAT. We cover this in Section \ref{sec:uhat}.

The second set of results concerns the expressivity of transformers with {\em averaging hard attention} mechanisms (a.k.a. Average Hard Attention Transformers, or simply AHAT). Such an attention mechanism --- which averages all values that maximize the attention score (unlike simply taking the leftmost value) --- provides another approximation of practical transformers, which use softmax attention. In particular, AHAT is tightly connected to Linear Temporal Logic extended with counting and the circuit complexity class $\TCzero$.
We cover this in Section \ref{sec:ahat}

Finally, we discuss the limitations of both UHAT and AHAT as approximations of practical transformers. In particular, we consider a recent promising direction that restricts AHAT to uniform attention layers (i.e., each position receives the same amount of attention). The resulting model, called AHAT[U], appears to be a good approximation of softmax transformers. We also discuss the distinction between expressibility and trainability in Section \ref{sec:limitations}.

\paragraph{Precision.} Real-world transformers are implemented on a specific
hardware that allows fixed (bit-)precision and fixed memory. Of course, one can
allow more precision and more memory by upgrading the hardware. Therefore, 
researchers in the theory of transformers has adopted a more practical approach
by specifying different precision model on a transformer $\transformer$: 
\begin{enumerate}
    \item \emph{Fixed} precision: there is a constant $c$ on the allowed 
        number of bits for any computation performed by $\transformer$.
    \item \emph{Logarithmic} precision: the number of allowed bits in the
        computation of $\transformer$ on a string of length $n$ is $O(\log n)$.
    \item \emph{Polynomial} precision: the number of allowed bits in the
        computation of $\transformer$ on a string of length $n$ is $O(n^c)$ for
        some constant $c$.
    \item \emph{Rational} (resp. \emph{real}) precision: this means rational
        (resp. real) computation is allowed with an unbounded precision.
\end{enumerate}
Although the distinction is important, it overcomplicates an introductory
article. For these reasons, we will assume the last precision model, and simply
remark
that all of the mentioned results work also for polynomial precision (and often
also logarithmic precision).

\paragraph{Notation and assumed background.} We assume familiarity with 
standard results in logic and automata, and their connections to circuit
complexity. All required background could be found in the excellent book
\cite{libkin-book} by Libkin.
In particular, we will consider {\em star-free} languages (i.e. regular
languages generated by regular expressions that use concatenation, union,
complementation, but no Kleene star), and their equivalent formulation using
first-order logic over strings (i.e. over the embedding of strings as logical 
structures, e.g., $aba$ is encoded as the structure with universe $\{1,2,3\}$,
the order relation $\preceq\ \subseteq \{1,2,3\}^2$, and unary relations $U_a =
\{1,3\}$ and $U_b = \{2\}$ indicating which positions labeled by $a$ and $b$,
respectively). By Kamp's theorem \cite{kamp}, the logic is equivalent to Linear
Temporal Logic (LTL). First-order logic characterization of star-free languages
can be extended with all numerical predicates to give us a characterization of
the circuit complexity class (nonuniform) $\ACzero$, which can be defined by a 
class of problems that can be solved by a family $\{C_n\}_{n \geq 0}$ of
constant-depth polynomial-sized (i.e. polynomial in $n$) boolean circuits (with 
unbounded fan-ins), wherein $C_n$ is employed to decide input strings of length 
$n$. Note that a $k$-ary numerical predicate simply means a relation $R
\subseteq \N^k$. In the sequel, we also use the fragment FO[Mon], which restricts the above use of numerical predicates only to \emph{monadic} (i.e. unary) numerical predicates. This is a strict subset of $\ACzero$.

The circuit complexity $\TCzero$ extends $\ACzero$ with majority gates,
which effectively allows one to encode all standard arithmetic operations on
numbers including addition, multiplication, etc. $\TCzero$ problems are often 
construed in the FLaNN (Formal Languages and Neural Networks) community as 
\emph{efficiently parallelizable} problems. Note that $\TCzero$ is a subset of
the circuit complexity class $\NCone$, which contains all problems solvable by
families of polynomial-sized circuits of logarithmic depth. It is known that
$\NCone$ contains all regular languages. [It is not known if all regular
languages are contained in $\TCzero$]. In turn, $\NCone$ is a subset of $\Logspace$, i.e., the class of problems solvable in logarithmic space.

\section{Formal Models of Transformers}
\label{sec:model}

We define several formal models of transformers, which are based on the type of
adopted attention mechanisms (i.e. hard or soft attention). We first define
these semantically, and then instantiate them based on different attention
mechanisms.

A transformer can be seen as a composition of several 
sequence-to-sequence transformations. More precisely, a \emph{seq-to-seq
transformation} is a length-preserving $f: (\R^l)^* \to (\R^h)^*$ for some
positive integers $l,h$. That is, $f$ maps an input sequence $\sigma$ of 
vectors of dimension $l$ to an output sequence $f(\sigma)$ of dimension $h$ of the
same length, i.e., $|f(\sigma)| = |\sigma|$. We write $\idim(f)$ (resp. $\odim(f)$) to 
denote the dimension of the input (resp. output) vectors of $f$, i.e., $l$
(resp. $h$). A sequence $\mu := f_1,\ldots,f_k$ of seq-to-seq transformers is said to
be \defn{well-typed} if $\idim(f_{i+1}) = \odim(f_i)$ for each $i=1,\ldots,k-1$.
We assume a finite \defn{alphabet} $\ialphabet$ of tokens (a.k.a. symbols or
characters) not containing the \defn{end-of-string symbol} $\EOS$. We write
$\ialphabet_\EOS$ to denote $\ialphabet \cup \{\EOS\}$.
A \emph{transformer} $\transformer$ over $\ialphabet$  can then 
be defined as a triple $(\mu,\tokenem,\vecT)$, where $\mu$ is a
well-typed sequence of seq-to-seq transformers as above, $\tokenem:
\ialphabet_\EOS \to \R^d$ with $d = \idim(f_1)$ is called a \emph{token 
embedding}, and $\vecT \in \R^{s}$ with $s = \odim(f_k)$. The token embedding
$\tokenem$ can be extended to $\tokenem: \ialphabet^* \to (\R^d)^*$ by 
morphism, i.e.,
$\tokenem(w_1\cdots w_n) = \tokenem(w_1)\cdots \tokenem(w_n)$, with $w_1\cdots
w_n \in \ialphabet^*$.
The language $L
\subseteq \ialphabet^*$ accepted by $\transformer$ consists precisely of 
strings $w \in \ialphabet^*$ such that the last vector $\vecV$ in 
\begin{equation}
    f_k(f_{k-1}(\cdots f_1(\tokenem(w\EOS)) \cdots )) \label{eq:accept}
\end{equation}
--- that is, at position $|w|+1$ in the sequence --- satisfies
$\langle \vecT,\vecV \rangle > 0$, where $\langle \vecT,\vecV \rangle$ denotes
the dot product of $\vecT$ and $\vecV$. That is, we first apply $f_1,\ldots,f_k$ (in
this order) to the sequence $\tokenem(w\EOS)$ of vectors, and check if a
weighted sum of the arguments in the last vector is positive.

\begin{remark}
    The above setting of transformers does not admit \emph{Chain of Thoughts
    (CoTs)}.  With CoTs, a transformer $\transformer$ on input $w$ will output 
    symbols, which are then continuously fed back into $\transformer$ until a specific
    output symbol is produced. That is, on input $w$, $\transformer$ produces 
    a symbol $a_1$. We then run $\transformer$ on input $wa_1$ and produce
    another symbol $a_2$, and so on. It is known that transformers with CoTs
    are Turing-complete \cite{turing1,turing2,MS24}. In the sequel, we do not consider transformers with CoTs. \qed
\end{remark}

We have thus far defined the notion of transformers only semantically. We now 
discuss how to define a seq-to-seq transformation more concretely. To this end,
we employ the following ideas:
\begin{enumerate}
    \item Use \emph{piecewise linear functions} to modify a single vector in
        the sequence. 
    \item Use \emph{attention} to ``aggregate'' several vectors in the sequence.
\end{enumerate}
We will discuss these in turn.

\subsection{Piecewise linear functions}
A \defn{piecewise linear function} is a function $f: \R^r \to \R^s$ that is
representable by a Feed-Forward Neural Network (FFNN). More precisely, a
piecewise linear function can be defined inductively:
\begin{description}
    \item[(Base)] Each identity function $Id: \R^r \to \R^r$ is piecewise
        linear.
    \item[(Affine)] If $f: \R^r \to \R^s$ is piecewise linear and $g: \R^s \to
        \R^t$ is an affine transformation\footnote{That is, given an input
        vector $\vecX$, we output $A\vecX + \vecC$, where $A$ is a linear 
        transformation and $\vecC$ is a constant vector.}, then the composition
        $f\circ g: \R^r \to \R^t$ is piecewise linear.
    \item[(ReLU)] If $f: \R^r \to \R^s$ is piecewise and $i \in\{1,\ldots,s\}$,
        then the function $g: \R^r \to \R^s$ defined as
        \[
            g(\vecV) = (w_1,\ldots,w_{i-1},\max\{0,w_i\},w_{i+1},\ldots,w_s), 
        \]
        where $f(\vecV) = (w_1,\ldots,w_s)$, is piecewise linear. 
\end{description}
As before, we can extend each piecewise linear function to sequences of vectors
by morphisms, i.e., $f: (\R^r)^* \to (\R^s)^*$ with $f(\vecV_1,\ldots,\vecV_n)
:= f(\vecV_1),\ldots,f(\vecV_n)$. Notice, however, such functions can
\emph{only}
modify a vector at the $i$th position in the sequence solely based on its values
and \emph{not} on the values of vectors at other positions. An intra-sequence 
aggregation of values is enabled by the so-called \emph{attention}, which we
discuss next.

\subsection{Attention layers}

To define an attention layer, we assume a \emph{weight normalizer} $\weight: 
\R^* \to \R^*$, which turns any $d$-sequence of weights into another $d$-sequence
of weights. We will define some common normalizers below, which will result in
hard and soft attention layers. 

A seq-to-seq transformation $f: (\R^r)^* \to (\R^s)^*$ generated by an
attention layer associated with $\weight$ is given by three piecewise linear 
functions $A,B,C$
\[
    A,B: \R^r \to \R^r \qquad C: \R^{2r} \to \R^s.
\]
defined as follows. On input $\sigma = \vecX_1,\ldots,\vecX_n$, we have
$f(\sigma) = \vecY_1,\ldots,\vecY_n$ such that
\[
    \vecY_i := C(\vecX_i,\vecV)
\]
where 
\begin{eqnarray}
    \vecV & := & \sum_{j=1}^n \vecW(j) \vecX_j, \label{eq:weighted_sum} \\
    \vecW & := & \weight(\{\langle A\vecX_i,B\vecX_j\rangle\}_{j=1}^n).
    \label{eq:weight}
\end{eqnarray}
In other words, an attention layer looks at a vector $\vecX_i$ at each position 
$i$ and decides ``how much attention'' is to be given to vectors
$\{\vecX_j\}_{j=1}^n$ at any
position in the input sequence. To this end, one obtains a sequence of weights 
$\{\langle \vecX_i,\vecX_j\rangle\}_{j=1}^n$. After normalizing this using
$\weight$, the result of the attention is $\vecV$, which is a weighted sum
$\{\vecX_j\}_{j=1}^n$ over all the input vectors.

\subsubsection*{Soft Attention.}
Practical transformers use weight normalizers defined by the softmax function,
which turns a sequence of weights into a probability distribution. In
particular, 
given a sequence $\sigma = x_1,\ldots,x_n \in \R^n$, define $\softmax(\sigma) 
:= y_1,\ldots,y_n$, where
\[
    y_i := \frac{e^{x_i}}{\sum_{j=1}^n e^{x_j}}.
\]
A \defn{SoftMax Attention Transformer (SMAT)} consists of seq-to-seq transformations that are
defined using the softmax weight normalizer.


\subsubsection*{Hard Attention}
As previously mentioned, softmax attention is sometimes rather difficult to
analyze, owing to the usage of exponential functions. This led researchers to
use other weight normalizers that led to the so-called \defn{hard attention
layers}. More precise, there are two common flavors: 
\defn{unique hard attention} and \defn{average hard attention}. A unique hard 
attention uses the weight normalizer $\uha$ that selects the leftmost maximum 
weight, i.e., $\uha(x_1,\ldots,x_n)=(y_1,\ldots,y_n)$, where
$y_i := 1$ if $i$ is the leftmost position in $\vecX := x_1,\ldots,x_n$ with
$x_i = \max(\vecX)$; or else $y_i := 0$. An average hard attention uses the 
weight normalizer $\aha$ that selects \emph{all} positions with maximum
weight, i.e., $\aha(x_1,\ldots,x_n)=(y_1,\ldots,y_n)$, where 
$y_i := 1/|P|$ if $x_i = \max(\vecX)$; or else $y_i := 0$. Here 
$P$ consists of positions $i$ in $\sigma$ such that $x_i$
            is maximum in $\vecX$.

\subsection{Positional information}
Thus far, we have actually defined a rather weak class of transformers (called
\emph{NoPE-transformers}) that cannot distinguish different positions in the 
input sequence. They recognize \emph{permutation-invariant} languages, i.e.,
a string $s$ is in the language $L$ iff all of the reorderings of $s$ are in
$L$. There are two common ways to recover ordering: (1) \emph{masking} and 
(2) \emph{Position Embeddings (PEs)}. We will go through these in turn.

\subsubsection*{Masking.} Masking is used to ``hide'' some positions in an
input sequence to a layer with respect to a certain ``anchor'' position. The
most commonly used type of masking is called \emph{strict future masking}, 
which we will focus on in the remainder of the paper.

Intuitively, when attention is applied with respect to the position $i$, we 
looked at \emph{all} positions and computed a normalized weight sequence
accordingly. The version with strict future masking modifies this by considering
only positions $j$ \emph{strictly before} $i$, i.e., $j < i$. Formally, one
simply modifies Equation \ref{eq:weighted_sum} and Equation \ref{eq:weight} by 
the masked version:
\begin{align*}
    \vecV &= \sum_{j=1}^{i-1} \vecW(j) \vecX_j, \qquad & 
    \vecW &= \weight(\{\langle A\vecX_i,B\vecX_j\rangle\}_{j=1}^{i-1}).
\end{align*}

\subsubsection*{Position Embeddings (PEs).} A \defn{Position Embedding} is an 
\emph{arbitrary} function of the form $p: \N \times \N \to \R^d$. The idea is 
that $p(i,n)$ indicates the position information of the vector at position $i$ for
a sequence of length $n$. Thus, to extend transformers by PEs, we first apply
both the token embedding and the PE $p$ to the input string $w=w_1\cdots w_n$
before processing the resulting sequence of vectors in the usual way. More 
formally, we modify the above acceptance condition in the definition of 
transformers by using 
\[
    f_k(f_{k-1}(\cdots f_1(\sigma)) \cdots ))
\]
where, instead of Equation \ref{eq:accept}, we use
\[
    \sigma :=
\tokenem(w_1)+p(0,n+1),\cdots,\tokenem(w_n)+p(n,n+1),\tokenem(\EOS)+p(n+1,n+1). 
\]

At this point, it is appropriate to ask what types of PEs are reasonable. In
practice, PEs may use trigonometric functions (e.g. $\sin$) and various other
information about the position in the sequence (e.g. the ``absolute'' position
$i$, the length $n$ of the sequence, etc.). Thus, researchers have studied
transformers with PEs \emph{without} any restriction whatsoever on the PEs.
Remarkably, some interesting results can already be proven in this setting. We
will mention some restricted classes later.
We end this section with an easy result:

\begin{proposition}
    Each Masked UHAT (resp. AHAT) with PEs can be simulated by UHAT (resp. 
    AHAT) with PEs with no masking.
\end{proposition}

\section{Unique Hard Attention Transformers}
\label{sec:uhat}

The first fundamental result concerning UHAT comes from \cite{Hahn20,HAF22},
which show that their class of languages is contained in the well-studied
circuit complexity class $\ACzero$, consisting of problems solvable by
constant-depth, polynomial-size Boolean circuits. More recently, this
containment was proven strict \cite{BKLP24}.

\begin{theorem}[\cite{Hahn20,HAF22,BKLP24}]
    UHAT with PEs is strictly subsumed in $\ACzero$.
    \label{th:ac0}
\end{theorem}

\begin{proof}[Proof idea]
    Let us quickly discuss the proof idea behind the containment in $\ACzero$. Fix
    an UHAT $\transformer$ with, say, $h$ layers.
    For simplicity, let us assume the alphabet $\ialphabet = \{a,b\}$. The key
    idea is that there is a polynomial function $p(n)$, for any possible 
    string length $n$, such that 
    the set $V_n$ of vectors --- as well as the set $S_n$ of
    possible 
    attention scores --- that can be generated in the computation of the UHAT 
    has size $|V_n| = O(p(n))$. More precisely, in the input layer after
    application of $\tokenem$ and position encoding, we can generate $O(n)$ many
    vectors. In the next layer, there are $O(n^2)$ many vectors. In the $k$th
    layer, there are $O(n^{2^{k-1}})$ possible vectors. Therefore, we may set 
    $p(n)$ to be $O(n^{2^h})$.

    Thus, we may represent each vector in $V_n$ and each attention score in 
    $S_n$ using $O(\log n)$ bits. Therefore, using a constant depth
    polynomial-sized boolean circuit (by a simple enumeration), we can 
    represent the relation $\preceq\ \subseteq S_n \times S_n$ containing 
    pairs $(s,s')$ such that $s$ has a smaller attention score as $s'$.
    Similarly, using a constant depth polynomial-sized boolean circuit, we can
    represent the relation $R \subseteq V_n \times V_n \times S_n$ such that
    $R(\vecV,\vecV',s)$ iff $\langle \vecV,\vecV'\rangle = s$. Together, this
    allows us to represent --- using constant depth polynomial-sized boolean
    circuits --- the function $f_\ell: V_n^n \times \{1,\ldots,n\} \to V_n$
    such that $f_\ell( (\vecV_1\cdots \vecV_n),i) ) = \vecV$ iff, whenever the
    $\ell$th layer has input sequence $\vecV_1\cdots \vecV_n$, the vector
    at position $i$ at layer \#$(\ell+1)$ is $\vecV$. All in all, this gives
    rise to a constant-depth polynomial-sized boolean circuit $C_n$ for input
    strings of length $n$.

    To conclude the theorem, we simply use the (non-uniform) family $\{C_n\}_{n
    \geq 0}$ of circuits to represent $\transformer$.
\end{proof}

Combined with well-known limitations of $\ACzero$ (e.g. see
\cite{Ajtai83,libkin-book}), the above result shows that some languages are not 
expressible by UHAT, including PARITY and MAJ, where the latter is defined as:
\[
    \text{MAJ} := \{ w \in \{a,b\}^* : |w|_a \geq |w|_b \}.
\]
While this provides us a ceiling of what languages 
are expressible as UHATs, the following
two results show what UHATs are capable of. To this end, we write FO[Mon] to
denote first-order logic over strings extended only by \emph{monadic} numerical
predicates (i.e. sets of numbers); recall that this would have yielded 
$\ACzero$ if extended with all $k$-ary ($k \geq 1$) numerical predicates; 
see \cite{libkin-book}. An example of monadic numerical predicates is $\Mod^3_2$
containing all numbers that are $2\pmod{3}$. 
\begin{theorem}[\cite{BKLP24,masked-uhat}]
    FO[Mon] is expressible by UHAT with PEs, as well as by masked UHATs with
    finite-image PEs. In addition, masked NoPE-UHAT
    coincides with FO, which in turn coincides with star-free languages. 
    \label{th:uhat_char}
\end{theorem}

\begin{proof}[Proof idea] 
To prove the containment of FO[Mon] in UHAT ---either with PEs or masked attention with finite image PEs--- we use Kamp's theorem \cite{kamp}: FO[Mon] coincides in expressivity with LTL[Mon], i.e., LTL formulas that also use monadic numerical predicates as atomic propositions. Unlike FO formulas, which have multiple variables, LTL formulas are \emph{unary}, meaning that their semantics is a set of positions over a string. This simpler structure of LTL aligns well with the expressive power of UHATs, allowing for a proof using structural induction. In particular, we inductively show that for every LTL formula $\phi$ with unary numerical predicates, there exists a UHAT ${\cal T}_\phi$ such that on input $\sigma = \vecX_1,\ldots,\vecX_n$, corresponding to the embedding of a word $w = a_1,\ldots,a_n \in \Sigma^+$, it outputs a sequence ${\cal T}_\phi(\sigma) = \vecY_1,\ldots,\vecY_n$ over $\{0,1\}$ that contains a 1 precisely in those positions of $w$ that satisfy $\phi$. 

Let us give some intuition on how to do the aforementioned induction proof. For the base case, we deal with only $Q_a$ (saying that the current position has letter $a$) or a monadic numerical predicate $U \subseteq \N$. We need to set up the token embedding function $\tokenem$ and position embedding $p$ with a large enough dimension so that information on truth/falsehood of each atomic proposition in the given LTL formula can be read off directly. For example, for a string $w := abaa$ with the LTL formula $\Global( \Mod^2_2 \to Q_a)$, we would map $w$ to the following sequence of vectors:
\[
    (1,0,0), (0,1,1), (1,0,0), (1,0,1)
\]
where the vector at position $i$ corresponds to    $(Q_a(i),Q_b(i),\Mod^2_2(i))$. Note, we omitted $\EOS$ and potentially other ``information'' in the PEs for readability.

For the inductive case, one introduces new arguments at each position (i.e. increases the dimension) to encode truth/falsehood of the formulas higher up in the parse tree. Note, we keep the information stored in the previous layer.

For boolean combinations, one can handle this with piecewise linear functions. That is, $\neg \varphi$ can be implemented by the function $1-x_{\varphi}$, where $x_{\varphi}$ encodes the value of $\varphi$ at the same position in the string. For $\varphi \vee \psi$, we can implement it as $x_{\varphi} + \ReLU(x_\psi - x_\varphi) =
x_{\varphi} + \max(0,x_\psi - x_\varphi)$.

We next give an intuition how to do $\F\varphi$ and show how to do this with
    PEs (with no masking). For other temporal operators, the reader is referred
    to \cite{BKLP24}.
    To this end, we assume by induction that the value
    $x_{\varphi}$ and $x_{\neg \varphi}$ are available at every position in the
    sequence. The first step is to ``nullify'' the value $x_{\neg \varphi}$ at
    the last position $n$, i.e., $x_{\neg \varphi}[n] := 0$. See the proof of
    Lemma 1 in \cite{BKLP24}. We then assume the use the following information
    at position $i$:
    \[
        \vecV_i := \langle \cos(\pi(1-2^{-i})/10),
        \sin(\pi(1-2^{-i})/10),1,x_{\neg\varphi} \rangle.
    \]
    With an appropriate affine transformation $B$, we have
    \[
        B \vecV_i := \langle \cos(\pi(1-2^{-i})/10),
        \sin(\pi(1-2^{-i})/10),-10.x_{\neg\varphi},0 \rangle.
    \]
    Thus, we have
    \[
        \langle \vecV_i,B \vecV_j \rangle := \cos(\pi(2^{-i}-2^{-j})/10) -
        10.x_{\neg\varphi}.
    \]
    The value $\cos(\pi(2^{-i}-2^{-j})/10)$ is maximized at position $j \geq 
    i$ and not at $j < i$.  In addition, the value $-10.x_{\neg\varphi}$ is
    maximized at $j=n$ (possibly also at $j < i$). Thus, it follows that
    $\langle \vecV_i,B\vecV_j\rangle$ is maximized at position $j \geq i$.
    Furthermore, it can be verified that among the value $j \geq i$ the value 
    $\cos(\pi(2^{-i}-2^{-j})/10)$ monotonically decreases in $j$. All in all,
    unique hard attention picks the vector $\vecV$ at the leftmost position $j \geq i$ 
    such that $w,j \models \varphi$ (otherwise, it picks the vector $\vecV$ at 
    position $n$), with which we can forward the truth/falsehood of $\F\varphi$.
\end{proof} 

\begin{corollary}
    UHAT with PEs contain all 
    regular languages expressible in $\ACzero$. 
    \label{cor:uhat_reg}
\end{corollary}
\begin{proof}
    Regular languages in $\ACzero$ are expressible in FO with unary numerical 
    predicates \cite{straubing-book} (more precisely, $\Mod^d_r$ containing all
    numbers 
    that are $r \pmod{d}$). The corollary then 
    follows from Theorem \ref{th:uhat_char}.
\end{proof}

Theorem \ref{th:uhat_char} turns out to be powerful enough to show the following
interesting ``non-regular'' capability of UHAT with PEs.
\begin{corollary}[\cite{BKLP24}]
    Palindrome is in UHAT with PEs.
\end{corollary}
\begin{proof}[Proof idea]
    Using PEs, it is possible to extend Theorem \ref{th:uhat_char} with any
    desired family $\{\preceq_n\}_{n \geq 0}$, where $\preceq_n$ deals with
    strings of length $n$. Therefore, on strings of length $n$, we could use 
    the ordering 
    \[
        1,n,2,n-2,3,n-3,\ldots
    \]
    of the set $\{1,\ldots,n\}$. This essentially turns the string $abccba$
    into $aabbcc$, for example. Therefore, using the unary numerical predicate $\Mod^2_1$, we can write an LTL[Mon] (or equivalent FO[Mon]) formula that
    says that at each odd position $i$ the next position $i+1$ has to have the 
    same label as that at position $i$.
\end{proof}

We conclude our discussion of UHAT by the problem of verifying Masked UHAT with
no PEs. By verifying, this could mean checking the emptiness, universality of
the language, or its equivalence to (or containment in) another Masked UHAT.
By Theorem \ref{th:uhat_char}, each Masked UHAT can be effectively turned into
a finite-state automaton recognizing the same language. Owing to decidability 
of emptiness, universality, equivalence, and containment for finite automata, we
obtain the same decidability results for Masked UHAT with no PEs.
\begin{corollary}[\cite{masked-uhat}]
    The problem of verifying Masked UHAT with no PEs is decidable.
\end{corollary}
Recently, Bergsträßer et al. \cite{succinct} has derived a precise complexity
for the problem, i.e., EXPSPACE-complete. This gives rise to a new interesting
challenge for verification.


\section{Logical Languages for Average Hard Attention}
\label{sec:ahat}

It is easy to construct an AHAT that recognizes MAJ. This takes AHAT beyond $\ACzero$. The following result shows that $\TCzero$ still upper-bounds the capability of AHAT.
\begin{theorem}[\cite{HAF22}]
    Languages recognized by AHAT are in $\TCzero$
\end{theorem}
The main reason behind the $\TCzero$ upper bound of AHAT is the ability of $\TCzero$-circuits to simulate arithmetic, which is needed in the computation of average hard attention.

For the time being no complete characterization for neither AHAT with PEs nor masked NoPE-AHAT exists. That is, we do not have an extension of Theorem \ref{th:uhat_char} to AHAT. However, it is still possible to specify a logic that expresses languages that can be expressed by AHATs. The logic is called \emph{Counting LTL}, as first defined in \cite{BKLP24}. Intuitively, Counting LTL extends LTL with linear counting terms of the form:
$$ 
C,C' \ := \ c\ (c \in \Z)\, \mid \, \LeftCount{\varphi} \, \mid \, \RightCount{\varphi} \, \mid \, C + C' \, \mid \, C - C', 
$$
and formulas of the form $C \leq C'$, where $C$ and $C'$ are linear counting terms. The term $\LeftCount{\varphi}$
(resp. $\RightCount{\varphi}$) counts the number of times $\varphi$ holds at positions before (resp. after) the one where we are evaluating the formula. The remaining terms and formulas have an intuitive meaning. 

We define the fragment $\CLTL$ of the Counting LTL, which removes all temporal operators of LTL, as well as terms of the form $\RightCount{\psi}$. That is, only terms of the form $\LeftCount{\varphi}$ is allowed. For instance, if $Q_a$ and $Q_b$ are formulas that check whether a position in a word holds symbol $a$ or $b$, respectively, then the $\CLTL$ formula 
$\LeftCount{Q_b} \leq \LeftCount{Q_a}$ checks whether the word belongs to $\text{MAJ}$ (if evaluated on the last position of the word). Similarly, we can define $\text{Dyck}$-1, the language of well-matched parenthesis words over the alphabet consisting of tokens $($ and $)$. The $\CLTL$ that checks for this language over the last position of a word in this alphabet is: 
$$ 
\LeftCount{Q_(} = \LeftCount{Q_)} \, \wedge \, \LeftCount{\LeftCount{Q_)} > \LeftCount{Q_(}} = 0,
$$
where we have used some standard logical abbreviations. It is possible to show
that the Counting LTL can express PARITY, whereas $\CLTL$ cannot express PARITY
\cite{len-gen-huang}.

\begin{theorem}[\cite{BKLP24,YC24}]
    Counting LTL extended with unary numerical predicates is in AHAT with PEs. 
    The fragment $\CLTL$ is expressible by masked NoPE-AHAT.
    \label{th:ahat_char}
\end{theorem}

The basic idea behind the proofs of these results is that AHATs allow to compute the uniform average
value among all positions that maximize the attention. This averaging mechanism allows to express many counting properties of interest. The proof is, again, by structural induction on Counting LTL formulas.

We showed that UHAT contains all regular languages in $\ACzero$. We do not know
if this is true for AHAT. That said, Theorem \ref{th:ahat_char} can be used to
show the following slightly weaker result:
\begin{corollary}[\cite{BKLP24}]
    If $\TCzero$ is strictly contained in $\NCone$, then AHAT with PEs contains
    all regular languages in $\TCzero$.
\end{corollary}

It turns out that, for the subclass of AHATs with no masking and no PEs, the 
following complete characterization can be proven:
\begin{theorem}[\cite{nope-ahat}]
    NoPE-AHAT recognizes precisely all permutation-invariant languages with semi-algebraic Parikh images.
    \label{th:nope-ahat}
\end{theorem}
To explain this theorem, recall (see \cite{KT10}) that the Parikh image 
$\Parikh$ of a language is a mapping from all strings in the language to their 
letter counts. For example, $\Parikh((ab)^*) = \{ (n,n) : n \in \N \}$. Here,
the tuple $(3,3)$ simply denotes that there are 3 occurrences of $a$s and 3
occurrences of $b$'s. Parikh's Theorem shows that context-free languages have
semilinear Parikh images, i.e., they are definable in Presburger Arithmetic. In
contrast, a relation $R \subseteq \N^k$ is \defn{semi-algebraic} if it can be
expressed as a finite union of nonnegative integer solutions to systems of
multivariate polynomial inequalities. That is, the above theorem implies (among
others) that languages $L_k$ of the form $\{ w \in \{a,b\}^* : |w|_a \geq 
(|w|_b)^k\}$,
are expressible by NoPE-AHAT; note that $L_k$ has no semilinear Parikh images
for $k \geq 2$. Interestingly, this also shows that Counting LTL does not
subsume NoPE-AHAT, since the former can only express permutation-invariant
languages with semilinear Parikh images \cite{nope-ahat}.

Theorem \ref{th:nope-ahat} yields immediately undecidability of verification of 
NoPE-AHAT since solvability of Diophantine equations is well-known to be 
undecidable \cite{Mat93}.
\begin{corollary}[\cite{nope-ahat}]
    Checking whether a NoPE-AHAT recognizes a nonempty language is undecidable.
\end{corollary}

\section{Limitations of UHATs and AHATs}
\label{sec:limitations}

Having gone through some body of results in the literature, we now discuss two main limitations of these results.

\subsubsection*{Limitation 1: Soft attention vs. Hard attention.}

As we remarked, practical transformers are based on soft attention.
It is still unclear whether the theory of expressivity of UHATs and AHATs
provides a good approximation of the theory of expressivity of softmax
transformers. For example, we do not know where the expressivity of softmax
transformers exactly lies (e.g. do they subsume UHATs?). That said, it is known
 that PARITY can be captured by a softmax transformer with PEs. Thus, softmax
transformers are not subsumed by UHATs \cite{CC22}. Furthermore, the relationship between
AHAT and softmax transformers has also not been fully delineated (for more on
this, see \cite{simulating,YC24}). 

One subclass of AHAT that seems to be a promising approximation of SMAT
restricts all layers to apply only \emph{uniform} attention. More precisely,
an AHAT layer is uniform if the piecewise linear functions $A,B: \R^r \to \R^r$
ensure that there exists a constant $c$ such that $\langle 
A\vecX,B\vecY\rangle = c$ for all $\vecX,\vecY \in \R^r$. This can happen
esp. when the linear transformation components of $A$ and $B$ map $\vecX$ and
$\vecY$ to 0. The subclass is denoted by AHAT[U]. The following result is
folklore and can easily be shown by noting that $\softmax(s_1,\ldots,s_n)=
\aha(s_1,\ldots,s_n) = [1/n \cdots 1/n]$, whenever $s_1 = \cdots = s_n$, which 
can be guaranteed for uniform AHAT layers.
\begin{proposition}
    Language recognized by AHAT[U] are also recognized by SMAT.
\end{proposition}
The above observation was already used in \cite{nope-ahat,YC24} to show the power of 
SMAT:
\begin{proposition}[\cite{YC24}]
    $\CLTL$ languages are recognizable by SMAT.    
\end{proposition}
\begin{proposition}[\cite{nope-ahat}]
    Permutation-invariant languages with semialgebraic Parikh images are 
    recognizable by SMAT.    
\end{proposition}

\subsubsection*{Limitation 2: Trainability vs. expressibility.} 

Not all expressible languages are efficiently trainable on transformers, i.e.,
by means of Stochastic Gradient Descent (SGD). This applies particularly to
PARITY \cite{CC22}, which seems to be extremely difficult to train on
transformers with any high enough level of accuracy, although it is expressible
by a softmax transformer. This phenomenon was very recently shown to be caused
by \emph{sensitivity}. Loosely speaking, PARITY is sensitive since flipping one
letter (i.e. $a$ to $b$ and vice versa) changes the parity of any string. 
Contrast this with MAJ,
where there are not so many strings that change their memberships in MAJ, after
flipping a letter. This was hypothesized to be the reason why MAJ is efficiently
trainable, whereas PARITY is not.

One interesting upshot of the research effort in understanding trainability is
the so-called \emph{RASP-L conjecture} \cite{raspl}, which states that a concept is likely to
length generalize (i.e. when trained on shorter strings, generalize to longer
strings) precisely whenever it is expressible as a short RASP-L program.
However, as noted by Huang et al. \cite{len-gen-huang}, this is not a precisely formulated
conjecture. The authors postulated a formal version of RASP-L conjecture by
replacing RASP-L with limit transformers and the logic $\CLTL$, for which they 
could successfully 
prove and empirically verify a length generalization theorem. In particular,
this ruled out PARITY (as it is not in $\CLTL$), but admits MAJ. It remains to
be seen if $\CLTL$ subsumes all concepts that admit length generalization on
transformers.


\section{Conclusions}
\label{sec:conc}

We have discussed several key results employing logic and automata for
understanding what is expressible in/efficiently trainable for transformers.
It must be emphasized that these are only a handful of results in this rapidly
growing field of FLaNN (Formal Languages and Neural Networks); for a more
comprehensive (though less detailed) account of FLaNN, see the excellent survey 
\cite{transformers-survey}. Such a survey is itself quickly becoming out of date
with many emerging topics (e.g. theory of transformers over time series
\cite{uhat_data}). It is our sincere hope that this article could
motivate more researchers in logic and automata, as well as verification and
programming languages, to take up some of the many pressing challenges in FLaNN.

\paragraph{Acknowledgment.} We thank Pascal Bergsträßer, David Chiang, Michael Hahn, Alexander Kozachinskiy, Andy Yang, and Georg Zetzsche for the fruitful discussion. Lin is supported by the European Research Council\footnote{https://doi.org/10.13039/100010663} under Grant No.~101089343 (LASD). Barcel\'o is funded by ANID - Millennium Science Initiative Program -
Code ICN17002, and by CENIA FB210017, Basal ANID.

%
%
\bibliographystyle{splncs04}
\bibliography{refs}

\begin{thebibliography}{10}
\providecommand{\url}[1]{\texttt{#1}}
\providecommand{\urlprefix}{URL }
\providecommand{\doi}[1]{https://doi.org/#1}

\bibitem{Ajtai83}
Ajtai, M.: {\(\sum\)}\({}^{\mbox{1}}\)\({}_{\mbox{1}}\)-formulae on finite
  structures. Ann. Pure Appl. Log.  \textbf{24}(1),  1--48 (1983)

\bibitem{BKLP24}
Barcel{\'{o}}, P., Kozachinskiy, A., Lin, A.W., Podolskii, V.V.: Logical
  languages accepted by transformer encoders with hard attention. In: ICLR
  (2024)

\bibitem{succinct}
Bergstr{\"{a}}{\ss}er, P., Cotterell, R., Lin, A.W.: Transformers are
  inherently succinct  (2025), manuscript under submission

\bibitem{uhat_data}
Bergstr{\"{a}}{\ss}er, P., K{\"{o}}cher, C., Lin, A.W., Zetzsche, G.: The power
  of hard attention transformers on data sequences: {A} formal language
  theoretic perspective. In: NeurIPS (2024)

\bibitem{turing2}
Bhattamishra, S., Patel, A., Goyal, N.: On the computational power of
  transformers and its implications in sequence modeling. In: CoNLL. pp.
  455--475 (2020)

\bibitem{CC22}
Chiang, D., Cholak, P.: Overcoming a theoretical limitation of self-attention.
  In: Muresan, S., Nakov, P., Villavicencio, A. (eds.) ACL. pp. 7654--7664
  (2022)

\bibitem{elman}
Elman, J.L.: Finding structure in time. Cognitive Science  \textbf{14}(2),
  179--211 (1990)

\bibitem{Hahn20}
Hahn, M.: Theoretical limitations of self-attention in neural sequence models.
  Trans. Assoc. Comput. Linguistics  \textbf{8},  156--171 (2020)

\bibitem{HAF22}
Hao, Y., Angluin, D., Frank, R.: Formal language recognition by hard attention
  transformers: Perspectives from circuit complexity. Trans. Assoc. Comput.
  Linguistics  \textbf{10},  800--810 (2022)

\bibitem{len-gen-huang}
Huang, X., Yang, A., Bhattamishra, S., Sarrof, Y.R., Krebs, A., Zhou, H.,
  Nakkiran, P., Hahn, M.: A formal framework for understanding length
  generalization in transformers. In: ICLR (2025)

\bibitem{kamp}
Kamp, H.W.: Tense Logic and the Theory of Linear Order. Ph.D. thesis,
  University of California, Los Angeles (1968)

\bibitem{KT10}
Kopczynski, E., To, A.W.: Parikh images of grammars: Complexity and
  applications. In: Proceedings of the 25th Annual {IEEE} Symposium on Logic in
  Computer Science, {LICS} 2010, 11-14 July 2010, Edinburgh, United Kingdom.
  pp. 80--89 (2010). \doi{10.1109/LICS.2010.21},
  \url{https://doi.org/10.1109/LICS.2010.21}

\bibitem{libkin-book}
Libkin, L.: Elements of Finite Model Theory. Springer (2004)

\bibitem{Mat93}
Matiyasevich, Y.V.: Hilbert's Tenth Problem. MIT Press, Cambridge,
  Massachusetts (1993)

\bibitem{MS24}
Merrill, W., Sabharwal, A.: The expressive power of transformers with chain of
  thought. In: ICLR (2024)

\bibitem{turing1}
P{\'{e}}rez, J., Barcel{\'{o}}, P., Marinkovic, J.: Attention is
  turing-complete. J. Mach. Learn. Res.  \textbf{22},  75:1--75:35 (2021)

\bibitem{nope-ahat}
S{\"{a}}lzer, M., K{\"{o}}cher, C., Kozachinskiy, A., Zetzsche, G., Lin, A.W.:
  The counting power of transformers. CoRR  \textbf{abs/2505.11199} (2025)

\bibitem{straubing-book}
Straubing, H.: Finite Automata, Formal Logic, and Circuit Complexity. Progress
  in Theoretical Computer Science, Birkhäuser Boston, MA, 1 edn. (1994).
  \doi{10.1007/978-1-4612-0289-9}

\bibitem{transformers-survey}
Strobl, L., Merrill, W., Weiss, G., Chiang, D., Angluin, D.: What formal
  languages can transformers express? {A} survey. Trans. Assoc. Comput.
  Linguistics  \textbf{12},  543--561 (2024)

\bibitem{YC24}
Yang, A., Chiang, D.: Counting like transformers: Compiling temporal counting
  logic into softmax transformers. CoRR  \textbf{abs/2404.04393} (2024)

\bibitem{masked-uhat}
Yang, A., Chiang, D., Angluin, D.: Masked hard-attention transformers recognize
  exactly the star-free languages. In: NeurIPS (2024)

\bibitem{simulating}
Yang, A., Strobl, L., Chiang, D., Angluin, D.: Simulating hard attention using
  soft attention. CoRR  \textbf{abs/2412.09925} (2024)

\bibitem{raspl}
Zhou, H., Bradley, A., Littwin, E., Razin, N., Saremi, O., Susskind, J.M.,
  Bengio, S., Nakkiran, P.: What algorithms can transformers learn? {A} study
  in length generalization. In: ICLR (2024)

\end{thebibliography}
%





\end{document}